\documentclass{amsart}

\usepackage{amsmath}
\usepackage{amsthm}
\usepackage{amsfonts}
\usepackage{amssymb}
\usepackage{dsfont}
\usepackage{stmaryrd}
\usepackage{xcolor}
\usepackage{algorithm}
\usepackage[noend]{algorithmic}

\usepackage{hyperref}
\usepackage[capitalize]{cleveref}

\usepackage{fullpage}
\usepackage{todonotes}

\newtheorem{example}{Example}

\newtheorem{proposition}{Proposition}
\newtheorem{lemma}{Lemma}
\newtheorem{theorem}{Theorem}
\newtheorem{corollary}{Corollary}
\theoremstyle{definition}
\newtheorem{definition}{Definition}

\newtheorem{notation}{Notation}

\DeclareMathOperator{\ord}{ord}

\DeclareMathOperator{\denom}{den}

\DeclareMathOperator{\dom}{lc}

\DeclareMathOperator{\Vect}{Vect}

\DeclareMathOperator{\pgcd}{gcd}
\DeclareMathOperator{\im}{Im}
\DeclareMathOperator{\re}{Re}
\def\hK{\hat{\mathbb K}}
\newcommand{\trsp}[1]{#1^\mathsf{T}} 

\title{Reduction-Based Creative Telescoping for\\ Definite Summation
of
D-finite Functions}
\author{Hadrien Brochet \and Bruno Salvy}
\dedicatory{Dedicated to the memory of Marko Petkov\v sek}
\begin{document}

\begin{abstract}
Creative telescoping is an algorithmic method initiated by
Zeilberger to compute definite sums
by synthesizing summands that telescope, called certificates.
We describe a creative telescoping algorithm that computes telescopers
for definite sums of D-finite functions as well as
the associated certificates in a compact form. The algorithm relies on
a discrete analogue of the generalized Hermite reduction, or
equivalently, a generalization of the Abramov-Petkov\v sek
reduction.
We provide a Maple implementation with good timings on a variety of
examples.
\end{abstract}
\maketitle

\section{Introduction}
The algorithmic computation of definite sums originates in
Zeilberger's
algorithm in the 1990's~\cite{Zeilberger1990b,Zeilberger1991a,WilfZeilberger1992a}.
Initially designed to deal with
hypergeometric sums, his method of \emph{creative telescoping} has
been extended to differential settings~\cite{AlmkvistZeilberger1990,Zeilberger1990,Takayama1990,Takayama1990a} and next generalized to
the large
class of \emph{D-finite functions} by
Chyzak~\cite{Chyzak-2000-EZF}.
In order to compute a definite sum of 
$F
(t,x_1,\dots,x_m)$ with respect to~$t$, where
each $x_i$ is a variable with respect to which one can 
apply a linear operator $\partial_i$ (generally, differentiation or
shift or
$q$-shift operator), the
creative telescoping algorithm 
constructs identities of the form
\begin{equation}\label{eq:telescoping}
\sum_{\boldsymbol\alpha}{c_{\boldsymbol\alpha}
(x_1,\dots,x_m)\boldsymbol\partial^
{\boldsymbol\alpha}(F)}=G(t+1,x_1,\dots,x_m)-G(t,x_1,\dots,x_m).
\end{equation}
Here, the sum is over a finite number of
multi-indices~$\boldsymbol\alpha$ and we use the multi-exponent
notation $\boldsymbol\partial^{\boldsymbol\alpha}=\partial_1^
{\alpha_1}\dotsm\partial_m^{\alpha_m}$. In the
original version for hypergeometric summation, the
monomials $\boldsymbol\partial^{\boldsymbol\alpha}(F)$ are simply
successive shifts~$F(t,n),F(t,n+1),F(t,n+2),\dots$
of a hypergeometric sequence
$F(t,n)$. Identities obtained that way
can often be summed over~$t$. The right-hand
side telescopes by design. Since the coefficients~$c_
{\boldsymbol\alpha}$ do not depend on the variable~$t$, the left-hand
side results in an operator
applied to
the definite sum of~$F$. From there, other algorithms can be
applied to compute information on the sum.
The operator in the left-hand
side of~\cref{eq:telescoping} is called a
\emph{telescoper} of~$F$ and the function~$G$ in the right-hand side
is the corresponding
\emph{certificate}. Chyzak's algorithm also deals
with the differential analogue of
\cref{eq:telescoping} where the right-hand side is a derivative; it is
used
to compute information on definite integrals. Chyzak's algorithm, like
Zeilberger's, looks for telescopers with an increasing number of
monomials~$\boldsymbol\partial^{\boldsymbol\alpha}$ with indeterminate
coefficients~$c_{\boldsymbol\alpha}$ and determines~$c_
{\boldsymbol\alpha}$ such that a certificate~$G$ exists in the vector
space generated by the~${\boldsymbol\partial}^{\boldsymbol\beta}(F)$
for ${\boldsymbol\beta}\in\mathbb N^{m+1}$ over the field of rational
functions. The conditions of being
\emph{D-finite} is that this vector space has finite dimension, which
allows for the existence of algorithms based on linear algebra.
If no certificate exists, then the support is
increased
and the process is
iterated. This stops either when sufficiently many operators have been
found or when a prescribed bound on the orders is reached. (In the
original hypergeometric case, no bound on the order is
fixed \emph{a priori} and termination is
guaranteed for the family of \emph{proper} hypergeometric terms~\cite{WilfZeilberger1992a}.)
Efficiency issues with this approach have led to the development of
heuristics and a very useful Mathematica implementation
by Koutschan~\cite{Koutschan2010}. 

The most recent approach to deal with the
efficiency issues with creative telescoping was initiated by 
Bostan, Chen, Chyzak and many co-authors who developed 
a class of \emph{reduction-based algorithms}
\cite{BostanChenChyzakLi2010,BostanChenChyzakLiXin2013a,ChenHuangKauersLi2015,BostanDumontSalvy2016,ChenKauersKoutschan2016,ChenHoeijKauersKoutschan2018}. These algorithms
avoid the computation
of potentially large certificates. In the differential case, where
the right-hand side of \cref{eq:telescoping} is replaced by a
derivative
$\partial_t(G)$, the
principle is to use a variant of Hermite reduction to compute an
additive decomposition of each
monomial in the form
\begin{equation}\label{eq:differential-reduction}
\boldsymbol\partial^{\boldsymbol\alpha}(F)=R_
{\boldsymbol\alpha}(t,x_1,\dots,x_m)F+\partial_t(G_
{\boldsymbol\alpha}),
\end{equation}
where $R_{\boldsymbol\alpha}$ is a rational function with a
certain minimality property. A telescoper is found by looking for
a linear dependency between these rational functions for a family of
monomials~$\boldsymbol\partial^{\boldsymbol\alpha}$. 
The computation of the rational function~$R_
{\boldsymbol\alpha}$ by Hermite reduction works by getting
rid of multiple poles and isolating a polynomial part. This was
first done for the integration of bivariate rational functions
\cite{BostanChenChyzakLi2010}, of hyperexponential functions 
\cite{BostanChenChyzakLiXin2013a} and of mixed
hypergeometric-hyperexponential functions 
\cite{BostanDumontSalvy2016}. In  these three cases, the vector space
generated by the functions ${\boldsymbol\partial}^{\boldsymbol\beta}( F)$
for ${\boldsymbol\beta}\in\mathbb N^2$ has dimension only~1 over the
rational functions. For summation, the analogous problem for bivariate
hypergeometric sequences was solved by replacing the Hermite reduction
by a modified Abramov-Petkov\v sek reduction, thereby providing a
faster variant of Zeilberger's algorithm~\cite{ChenHuangKauersLi2015}.
For bivariate problems of dimension larger than~1, the method was
extended to the
integration
of bivariate algebraic functions \cite{ChenKauersKoutschan2016}, of
Fuchsian functions \cite{ChenHoeijKauersKoutschan2018} and more
recently of P-recursive sequences~\cite{ChenDuKauersWang2023,Du2023}
by means of
suitable integral bases. An extension 
to the integration of purely differential bivariate 
D-finite functions in
arbitrary dimension was first
achieved by turning the
differential equations satisfied by the function to be integrated 
into first-order differential systems; then, a variant of Hermite
reduction can be designed at the level of vectors of rational
functions~\cite{Van-Der-Hoeven2017}. This approach generalizes to
purely
differential D-finite functions in more variables~\cite{Hoeven2021}.

Another method relies on cyclic vectors and
allows the integration of arbitrary D-finite 
functions~\cite{Bostan_2018}. Without loss of generality, we assume that~$F$ is
a cyclic vector for $\partial_t$, which means that all monomials
$\boldsymbol\partial^
{\boldsymbol\alpha}(F)$ rewrite as $M_{\boldsymbol\alpha}(F)$ with
$M_{\boldsymbol\alpha}$ a
linear operator
in~$\partial_t$ only. (If $F$ is not a cyclic vector, one finds a
cyclic vector~$G$, $F=M_F(G)$ for some linear operator~$M_F$ in
$\partial_t$ only and the
rest of the reasoning is unchanged.) Next, for any rational
function~$u$ and any
linear operator~$M$ in~$\partial_t$,
repeated integration by parts
implies Lagrange's identity
\[uM(F)-M^*(u)F=\partial_t(P_M(F,u)),\]
where $M^*$ is the adjoint of $M$ and $P_M$ is linear in
$F,\partial_t(F),\dots$
and $u,\partial_t(u),\dots$. Thus, a reduction-based algorithm is
obtained by the additive decomposition of 
\cref{eq:differential-reduction} with $R_{\boldsymbol{\alpha}}$ a
solution of the \emph{generalized Hermite reduction}
\[R_{\boldsymbol{\alpha}}=M_{\boldsymbol{\alpha}}^*(1)\bmod
\operatorname{Im}(L^*).\]
The tools used in the reduction 
modulo the image of the linear differential operator $L^*$ are
classical techniques used when
looking for rational solutions of linear differential equations.

This method of integration using generalized Hermite reduction
extends to other contexts.
This has been done in high generality in a preprint by
van~der~Hoeven~\cite{vanderhoeven:hal-01773137}. Our approach here is
different: we focus on the case of summation only and give a
simple self-contained presentation of the corresponding algorithm; 
our algorithm returns operators of minimal order\footnote{Remark~5.6
in~\cite{vanderhoeven:hal-01773137} seems to allude to a way of
doing this, but the relevant space~$E$ may contain rational functions that are not in the image of $L$. For example take $L=1/z + 1/(z-1)\sigma^{-1}$, $\alpha = 0$, and $A =\{\alpha\}$, then one can check that $1/z \notin \im (L)$ but $1/z \in E$.}; we
make the
choice to avoid algebraic extensions when possible; we
present a Maple implementation that performs well in
practice. Note that while in terms of complexity, minimal operators
cannot be computed in polynomial time in general, in practice this
does not seem to be an obstacle.

\section{Example}
The multiplication theorem for Bessel functions of the first kind
$J_\nu$ states that~\cite[10.23.1]{OlverLozierBoisvertClark2010}
\begin{equation}\label{eq:multBessel}
J_\nu(\lambda z)=\lambda^\nu \sum_{n=0}^\infty \frac{(-1)^n(\lambda^2-1)^n(z/2)^n}{n!}J_{\nu+n}(z).
\end{equation}
This can be proved automatically by showing 
that the left-hand side and the right-hand side satisfy the same set
of mixed differential-difference equations with sufficiently many
identical initial conditions. 

We write $F$ for the summand in \cref{eq:multBessel}. It is a function
of the four variables $\nu,n,z,\lambda$. Basic properties of the
Bessel function give the following four equations:
\begin{align}
(\lambda^2-1)zS_\nu(F) + 2(n+1)S_n(F)  &= 0, \label{intro-exple-Snu} \\
(\lambda^2-1)\partial_\lambda(F) -2n\lambda F &= 0, \label{intro-exple-dlambda} \\ 
(1-\lambda^2)z\partial_z(F) + +2(n+1)S_n(F)
 + (\lambda^2-1)(2n+\nu)F &=
0, \label{intro-exple-dz}\\ 
4(n+1)(n+2)S_n^2(F) + 4(\lambda^2-1)(n+1)(n+\nu+1)S_n(F) + z^2
(\lambda^2-1)^2F &= 0, \label{intro-ex-Sn}
\end{align}
where $S_\nu$ denotes the shift with respect to~$\nu$: $S_\nu:G
(\nu)\mapsto G(\nu+1)$ and similarly for~$S_n$, while $\partial_z$
and $\partial_\lambda$ denote partial derivatives.
These equations show that any shift or derivative
$S_\nu^a\partial_\lambda^b\partial_z^cS_n^dF$ of $F$ with nonnegative
integers~$a,b,c,d$
rewrites
as a $\mathbb Q(\nu,\lambda,z,n)$-linear combination of $F$ and $S_n
(F)$. In particular, this implies that $F$
is D-finite with respect to these variables.
The aim of creative telescoping is to find a similar set of equations,
in the variables $\nu,\lambda,z$ only, for the sum in \cref{eq:multBessel}.

Let $\Delta_n$ be the \emph{difference operator} $\Delta_n=S_n-1$. Any
product $\phi(n)S_nF$ with $\phi\in\mathbb Q(\nu,\lambda,z,n)$ can
be
rewritten $\phi(n-1)F+\Delta_n(\phi(n-1)F)$, i.e., as the sum of a
rational function times~$F$ plus a difference, that would telescope
under summation. Consequently, any  $S_\nu^a\partial_\lambda^b\partial_z^cS_n^dF$
 with nonnegative
integers~$a,b,c,d$
rewrites in the form
\begin{equation}\label{eq:reduction}
\boldsymbol\partial^{\boldsymbol\alpha}(F)=R_
{\boldsymbol\alpha}F+\Delta_n(G_
{\boldsymbol\alpha}),
\end{equation}
with $R_{\boldsymbol\alpha}$ a rational function in $\mathbb Q
(\nu,\lambda,z,n)$. This reduction as a sum of a rational function
plus a difference is a general
phenomenon (see \cref{sec:cyclic}). 
Reduction-based creative telescoping works by reducing this
rational function further by pulling out parts that can be
incorporated into the difference $\Delta_n(G_{\boldsymbol\alpha})$. Denote by
$L$ the recurrence operator such that \cref{intro-ex-Sn} is $LF=0$.
The \emph{adjoint} of $L$ (see \cref{Lagrange}) is
\[
L^* = 4(n-1)nS_n^{-2} + 4(\lambda^2-1)n(n+\nu)S_n^{-1} + z^2(\lambda^2-1)^2
\]
where $S_n^{-1}:g(n)\mapsto g(n-1)$. \Cref{propCT} shows that a
rational
function $R$ is of the form $\Delta_n(M(F))$ for a
recurrence operator~$M(S_n)$ if and only if~$R$ is in the image $L^*(\mathbb Q
(\nu,\lambda,z,n))$. This is the basis for the computation of
relations of the form~\eqref{eq:reduction}
where $R_{\boldsymbol\alpha}$  is now a \emph{reduced} rational
function (in a sense made precise in \cref{CanForm}).
The next step is to look for linear combinations of these rational
functions that yield telescopers.

The starting point is the monomial $1$, which decomposes as
\begin{equation}
\label{exintro:cbl1}
 1\cdot F = 1\cdot F + \Delta_n(0).
\end{equation}
Using \cref{intro-exple-dlambda}, the monomial $\partial_\lambda$
rewrites
\begin{equation}
\label{exintro:cbl2}
\partial_\lambda(F) = \frac{2n\lambda}{\lambda^2-1}F + \Delta_n(0)
\end{equation}
and the rational function is reduced. Taking
the derivative of this equation and using  
\cref{intro-exple-dlambda} again gives a similar equation for
$\partial_\lambda^2(F)$:
\begin{align*}
\partial_\lambda^2(F) &= -\frac{2n(\lambda^2+1)}{(\lambda^2-1)^2}F + \frac{2n\lambda}{\lambda^2-1}\partial_\lambda(F) + \Delta_n(0), \\
&= -\frac{2n(\lambda^2+1)}{(\lambda^2-1)^2}F + \frac{(2n\lambda)^2}{(\lambda^2-1)^2}F + \Delta_n(0).
\end{align*}
This time, a reduction is possible. Indeed, \cref{propCT}
implies that $L^*(1)F$ is a difference~$\Delta(A_n)$ (where $A_n$ can
be computed explicitly). Since
\[
L^*(1) = 4\lambda^2n^2 + 4((\lambda^2-1)\nu - 1)n + z^2(\lambda^2-1)^2,
\]
we can eliminate the term in $n^2$  in the expression of
$\partial_\lambda^2(F)$ to get
\begin{equation}
\label{exintro:cbl3}
\partial_\lambda^2(F) = -\frac{2n(2\nu+1)}{(\lambda^2 - 1)}F -z^2F +
\Delta_n\left(\frac{A_n}{(\lambda^2-1)^2}\right).
\end{equation}
A simple linear combination of 
\cref{exintro:cbl1,exintro:cbl2,exintro:cbl3} then eliminates the term
in~$n$, showing
 that $F$ satisfies the equation
\[
\lambda\partial_\lambda^2F + (2\nu + 1)\partial_\lambda F + \lambda
z^2F = \Delta_n\left(\frac{-\lambda A_n}{(\lambda^2-1)^2}\right).
\]
The left-hand side is a telescoper. The right-hand side is a
certificate. It can be written more explicitly as
\[
\frac{-\lambda A_n}{(\lambda^2-1)^2}= - \frac{4(n\lambda^2 +
\lambda^2\nu - \nu - 1)n\lambda}{(\lambda^2 - 1)^2}F - \frac{4\lambda(n + 1)n}{(\lambda^2 - 1)^2} S_n(F).
\]
In general, summation and telescoping of the certificate requires 
verification. Here, we first observe that the certificate does not
have
integer poles and thus is well defined at all points over which it is
summed. Next, the certificate evaluates to zero at $n=0$. Finally,
it tends to zero when $n$ tends to infinity, as $J_{\nu+n}(z)$
decreases fast as $n\rightarrow\infty$ 
\cite[10.19.1]{OlverLozierBoisvertClark2010}.

In summary, we have obtained that the sum $S$ in the right-hand side
of \cref{eq:multBessel} satisfies
\[
\lambda\partial_\lambda^2(S) + (2\nu + 1)\partial_\lambda(S) + \lambda z^2S = 0.
\]
Proceeding similarly with \cref{intro-exple-Snu,intro-exple-dz}, one
gets the equations
\[
z\lambda S_\nu(S) + \partial_\lambda(S) = 0,\quad
z\partial_z(S) - \lambda\partial_\lambda(S) -\nu S =0.
\]
Injecting $T=J_\nu(\lambda z)/\lambda^\nu$ in these equations and
using
basic equations for~$J_\nu$ shows that it is a solution of this
system too. The proof of the multiplication theorem is concluded by
checking the equality of the initial conditions for $T$ and
for the sum on the
right-hand side of 
\cref{eq:multBessel}. As $\nu$ is associated to the shift, we need to
check initial conditions for any $\nu$ satisfying $0\leq\re(\nu) < 1$.
Indeed, both term of the identity equal $J_\nu(1)$ at $z=1,\lambda
=1$, and $\nu\in
[0,1)$ and both their derivatives
with
respect to~$\lambda$ equal $-J_{\nu+1}(1)$,  which proves the identity.

\section{Background}
In this section, we recall the basic framework for reducing
creative telescoping to the generalized Abramov-Petkov\v sek
decomposition. Most of this section is identical to the differential
case~\cite[Sec.~4]{Bostan_2018}, except for the existence
and computation of the
cyclic vector and the use of the recurrence
variant of Lagrange's identity~\cite{BarrettDristy}.
More gentle
introductions to Ore algebras,
creative
telescoping and their applications can be found in the references~\cite{Chyzak2014,ChyzakSalvy1998}.

\subsection{Telescoping ideal}
\subsubsection*{Ore algebras}
Let $\mathbf k$ be a field of characteristic~0, $x_0,\dots,x_m$ be
variables used to form the fields of rational functions $\mathbb
K=\mathbf k(x_1,\dots,x_m)$ and $\hK=\mathbb K(x_0)$. The \emph{Ore algebra}
$\mathbb A=\hK\langle \partial_0,\dots,\partial_m\rangle$
is a polynomial ring
over~$\hK$, with
$\partial_i\partial_j=\partial_j\partial_i$,
and a commutation between the $\partial_i$s and the elements
of~$\hK$ ruled by relations
\begin{equation}\label{eq:commutation}
\partial_i R=\sigma_i(R)\partial_i+\delta_i(R),\qquad R\in\hat{\mathbb
K},
\end{equation}
with $\sigma_i$ a ring morphism of~$\hK$ and $\delta_i$ a
$\sigma_i$-derivation, which means that $\delta_i(ab)=\sigma_i
(a)\delta_i(b)+\delta_i(a)b$ for all $a,b$ in~$\hK$
\cite{BronsteinPetkovsek1996,ChyzakSalvy1998}. The typical cases are when
$\partial_i$ is the differentiation $d/dx_i$ (then $\sigma_i$ is
the identity and $\delta_i=d/dx_i$) and the shift
operator~$x_i\mapsto x_i+1$ (then $\sigma_i(a)=a|_{x_i\leftarrow
x_i+1}$ and $\delta_i=0$).

\subsubsection*{Annihilating and D-finite ideals}
For a given function~$f$ in a left $\mathbb A$-module, the
annihilating ideal of~$f$ is the left ideal
$\operatorname{ann}f\subseteq\mathbb A$ of elements of~$\mathbb 
A$ that annihilate~$f$. 
A left ideal~$\mathcal I$ of~$\mathbb A$ is \emph{D-finite} when the
quotient $\mathbb{A}/\mathcal I$ is a finite dimensional $
\mathbb{K}$-vector space. A function is called D-finite when its
annihilating ideal is D-finite.

\subsubsection*{Telescoping ideal}
As we focus here on summation, from now on, when we use~$n$ and $S_n$,
they stand for $x_0$ and the corresponding shift
operator~$\partial_0:x_0\mapsto x_0+1$.

The \emph{telescoping ideal} $\mathcal{T}_{\mathcal I}$
of the left ideal~$\mathcal
I\subset\mathbb A$ with respect to~$n$ is
\begin{equation*}
\mathcal{T}_{\mathcal I} =\left(\mathcal I + \Delta_n(
\mathbb{A})\right) \cap {\mathbb{K}}\langle
\partial_1,\dots,\partial_m\rangle,\qquad \text{where}\quad
\Delta_n=S_n-1.
\end{equation*}
In other words, if $\mathcal I=\operatorname{ann}F$, 
the telescoping ideal $\mathcal T_\mathcal I$ is the set of
operators $T\in
\mathbb{K}\langle
\partial_1,\dots,\partial_m\rangle$ such that there exists $G\in
\mathbb{A}$ such
that $T + \Delta_nG \in \mathcal I$, or equivalently, such that 
\cref{eq:telescoping} holds (with $t=n$).

\subsection{Cyclic vector and Lagrange identity}\label{sec:cyclic}
\subsubsection*{Cyclic vector}
Let $\mathcal I$ be a D-finite ideal of $\mathbb A$ and let $r$ be the
dimension of the $\hK$-vector space $\mathbb B:=
\mathbb{A}/\mathcal I$. An
element~$\gamma\in \mathbb B$ is called \emph{cyclic
with respect to~$\partial_0$} if $
\{\gamma,\dots,\partial_0^{r-1}\gamma\}$ is a basis of
$\mathbb B$. In the differential case ($\partial_0=d/dx_0$), such a
vector always exists and can be computed efficiently when $\mathcal I$
is D-finite~\cite{ChurchillKovacic2002}. In the shift case 
($\partial_0:x_0\mapsto x_0+1$), even for a D-finite ideal~$\mathcal
I$, it is not
the case that there always exists a cyclic vector: in general,
$\mathbb B$ decomposes as the sum of a vector space where~$\partial_0$
is nilpotent and a part where it is cyclic~\cite{Jacobson1937}.
However, we have the following.
\begin{proposition}\cite[Thm.~B2]{HendricksSinger1999} With the
notation above, in the case when~$\partial_0$ is the shift operator
$x_0\mapsto x_0+1$, let $E=\trsp{(e_1,\dots,e_r)}$ be a basis of the
vector
space $\mathbb B=\mathbb A/\mathcal I$ and $A_0\in\hat{\mathbb K}^
{r\times r}$ be defined by
$\partial_0E=A_0E$. If $A_0$ is invertible, then there
exists a
cyclic vector with respect to~$\partial_0$ of the
form~$v=a_1e_1+\dots+a_re_r$ with polynomial
coefficients~$a_i\in\mathbb Z[x_0]$ of degree at most~$r-1$, and
coefficients all in~$\{0,\dots,r\}$.
\end{proposition}
Sufficient conditions for the matrix~$A_0$ to be invertible are that
$\mathcal I=\operatorname{ann}f$ with $f$ in a $\hK
[\partial_0,\partial_0^{-1}]$-module~\cite{HendricksSinger1999} or
that $\mathcal I$ be a reflexive 
ideal~\cite{vanderhoeven:hal-01773137}. In practice, this condition
on~$A_0$ can be checked from the input and appears to be always
satisfied in the examples we have tried. From this proposition, the
computation of a cyclic vector
follows the same lines as that of the differential 
case~\cite{ChurchillKovacic2002}. Most often, $e_1=1$ is a cyclic
vector, which simplifies the rest of the computation.

\subsubsection*{Lagrange's identity}
For our purpose, the shift version of Lagrange's identity can be
viewed as giving an
explicit form of the result of the left Euclidean
division by the difference operator~$\Delta_n$, when applied to a
left multiplication
by a rational function.
\begin{proposition} \cite{BarrettDristy}
\label{Lagrange}
Let $u\in\hat{\mathbb K}$ and let $L=\sum_
{i=0}^ra_iS_n^i$
be an operator of order $r$ with $a_i$ in $\hat{\mathbb{K}}$. The
adjoint operator $L^*$ of $L$ is defined as
$L^*=\sum_{i=0}^{r} a_i(n-i)S_n^{-i}$ and it satisfies
\begin{equation}\label{eq:Lagrange}
	uL-L^*(u)=\Delta_n P_L(u)
\end{equation}
where
\begin{equation}\label{eq:bilinearconcomitant}
P_L(u(n))=\sum_{i=0}^{r-1}\left(\sum_{j=i+1}^r{a_j
(n+i-j)u(n+i-j)}\right)S_n^i.\end{equation}
\end{proposition}
Note that the term $L^*(u)$ denotes the evaluation of the operator
$L^*$ at the rational function $u$, rather than the product of $L^*$ by $u$.

Let $\gamma$ be a cyclic vector. Then any element of~$\mathbb B$
is
of the form $A\gamma$ with $A\in\hK\langle S_n\rangle$.
Applying Lagrange's identity with~$u=1$, $L=A$ and multiplying on the
right by~$\gamma$ shows
that this is a
rational multiple of~$\gamma$ up to a difference:
\begin{equation}\label{eq:basic-cert}
A\gamma=A^*(1)\gamma+\Delta_n P_A(1) \gamma.
\end{equation}

As in the differential
case, 
all computations in~$\mathbb B$ then reduce to $\hK$-linear
operations on
single rational
functions, rather than vectors of them, by the following analogue 
of~\cite[Prop.~4.2]{Bostan_2018}.
\begin{proposition} 
\label{propCT}
With the notation above, let $\gamma$ be a cyclic
vector of~$\mathbb B=\mathbb A/\mathcal I$ and for
all~$i=0,\dots,m$, let $B_i\in\hK\langle S_n\rangle$ be such that $\partial_i\gamma=B_i\gamma$. Then
for all $R\in\hK$,
\begin{gather}\label{eq:def-varphii}
\partial_iR\gamma=\varphi_i(R)\gamma + \Delta_n Q_i(R)\gamma,\\
\text{with}\quad
\begin{cases}
	\varphi_i(R)=B_i^*(R(x_i+1)),\quad Q_i(R)=P_{B_i}(R
	(x_i+1))&
	\text{if $\partial_i:x_i\mapsto x_i+1$;}\\
	\varphi_i(R)=B_i^*(R) + \frac{d}{dx_i}(R),\quad 
	Q_i(R)=P_{B_i}(R
	(x_i))&\text{if $\partial_i=d/d{x_i}$.}
\end{cases}\notag
\end{gather}
\end{proposition}
\begin{proof}Multiplying \cref{eq:commutation} by $\gamma$ on the
right and using the definition of~$B_i$ gives
\[\partial_iR\gamma=\sigma_i(R)B_i\gamma+\delta_i(R)\gamma.\]
The conclusion follows from Lagrange's identity~\eqref{eq:Lagrange}
applied with $L=B_i$ and $u=\sigma_i(R)$.
\end{proof}

\subsection{Canonical Form}\label{sec:canonical}
\Cref{propCT} shows how, given a cyclic vector~$\gamma$, all elements
of~$\mathbb B$ can be reduced to the product of~$\gamma$ by a rational
function, up to a difference in~$\Delta_n\mathbb B$. The starting
point of the reduction-based creative telescoping is that one can
actually identify those multiples that belong to~$\Delta_n\mathbb B$.
\begin{proposition}
\label{propCT2}
With the same hypotheses as in \cref{propCT}, let $L$ be a minimal-order
operator in $\hK\langle S_n \rangle$ annihilating $\gamma$,
ie, the product $L\gamma$ is~0 in~$\mathbb A/\mathcal I$, $L$ has
order~$r$ and no operator of order~$r-1$ has that property.
Then for all $R\in\hK$,
$
R\gamma\in\Delta_n(\mathbb B) \Longleftrightarrow R\in L^*(
\hK).
$
\end{proposition}
\begin{proof}
First, if $R=L^*(R')$ with
$R'\in
\hK$, Lagrange's
identity \eqref{eq:Lagrange} with $u=R'$ and $L\gamma=0$ implies that 
 $R\gamma=L^*(R')\gamma=\Delta_n G\gamma$ for~$G=-P_L(R')$.
Conversely, if $R\gamma \in \Delta_n(\mathbb B)$, there exists 
$M\in\hK\langle S_n \rangle$ such that $R\gamma = \Delta_n M
\gamma$.
The operator $\Delta_n M - R$ annihilates $\gamma$. By
minimality of $L$, 
there exists $N\in\hK\langle S_n \rangle$ such that 
$\Delta_n M - R=NL$. Taking the adjoint and evaluating at~1 gives
$R=M^*\Delta_n^* - L^*N^*$ and finally $R=-L^*(N^*(1))$.
\end{proof}
This proposition motivates the following.

\begin{definition}\cite{Bostan_2018}
\label{CanForm}
A \emph{canonical form} associated to $L^*$ is a ${\mathbb{K}}$-linear
map $[\cdot] : \mathbb{K}(n) \rightarrow \mathbb{K}(n)$ such that for
all~$R\in\mathbb K(n)$,
$[L^{*}(R)]=0$ and $R-[R]\in L^{*}(\mathbb{K}(n))$.
A rational function $R\in\mathbb K(n)$ is called \emph{reduced} when $
[R]=R$.
\end{definition}
The computation of canonical forms is the object of \cref{sec:reduction}.

\subsection{Creative Telescoping Algorithm via Canonical Forms\label{subsec:CTalg}}
\begin{algorithm}[t]
\caption{Creative Telescoping Algorithm}
\label{alg:myalg}
\begin{algorithmic}
\label{AlgoCT}
\REQUIRE Generators of a D-finite ideal $\mathcal I$
\ENSURE A basis of the telescoping ideal $\mathcal{T}_{\mathcal I}$ of
$\mathcal I$
\STATE $\gamma \gets$ a cyclic vector of $\mathcal B:=
\mathbb{A}/\mathcal I$ with respect to  $S_n$\hfill\COMMENT{See 
\cref{sec:cyclic}}
\STATE $ L \gets $ a minimal order operator in $S_n$ annihilating
$\gamma$\hfill\COMMENT{See \cref{sec:canonical}}
\STATE $\varphi_1,\dots,\varphi_n$ the maps described in \cref{propCT}
\STATE Initialize CanonicalForm  \hfill\COMMENT{See \cref{sec:reduction}}
\STATE $R_{\boldsymbol{0}} \gets$ CanonicalForm($A^*_1
(1),L^*)$\hfill\COMMENT{See \cref{subsec:CTalg}}
\STATE $\mathcal{L} \gets [1]$  \hfill \COMMENT{List of monomials in $\partial_1,\dots,\partial_m$ to visit}
\STATE $\mathcal{G} \gets \{\}$ \hfill \COMMENT{Gröbner basis}
\STATE $\mathcal{Q} \gets \{\}$ \hfill \COMMENT{Basis of the quotient}
\STATE $\mathcal{R} \gets \{\}$ \hfill \COMMENT{Set of reducible
monomials} \WHILE{$\mathcal{L}\neq \emptyset$}
\STATE Remove the first element $\boldsymbol{\partial}^{
\boldsymbol{\alpha}}$ of $\mathcal{L}$
\IF{ $\boldsymbol{\partial}^{\boldsymbol{\alpha}}$ is
not a multiple of an element of $\mathcal{R}$}
\IF{ $\boldsymbol{\partial}^{\boldsymbol{\alpha}}\neq 1$}
\STATE Take $i$ such that $\boldsymbol{\partial}^{\boldsymbol{\alpha}}/ \partial_i\in\mathcal{Q}$;
 $R_{\boldsymbol{\alpha}} \gets$ CanonicalForm($\varphi_j(R_{
\boldsymbol{\partial}^{\boldsymbol{\alpha}}/\partial_i}),L^*)$
\ENDIF
\IF{there exists a  ${\mathbb{K}}$-linear relation between $R_
{\boldsymbol{\alpha}}$ and $\{R_{\boldsymbol{\beta}} \mid
R_{\boldsymbol{\beta}} \in \mathcal{Q}\}$}
\STATE $(\lambda_{\boldsymbol{\beta}})_{R_{
\boldsymbol{\beta}}\in\mathcal{Q}} \gets$ coefficients of
the relation  $R_{\boldsymbol{\alpha}}=\sum_{R_{\boldsymbol{\beta}}\in\mathcal{Q}} \lambda_{
\boldsymbol{\beta}}R_{\boldsymbol{\beta}}$
\STATE Add ${\boldsymbol{\partial}}^{\boldsymbol{\alpha}}- \sum_
{R_{\boldsymbol{\beta}}\in
\mathcal{Q}}
\lambda_{\boldsymbol{\beta}}\boldsymbol{\partial}^{\boldsymbol{\beta}}$ to $
\mathcal{G}$;
Add $\boldsymbol{\partial}^{\boldsymbol{\alpha}}$ to $\mathcal{R}$
\ELSE 
\STATE Add $\boldsymbol{\partial}^{\boldsymbol{\alpha}}$ to $\mathcal{Q}$
\STATE\algorithmicfor\ {$j=1$ \TO $m$}\ \algorithmicdo\ 
{Append the monomial $\partial_j\boldsymbol{\partial}^{
\boldsymbol{\alpha}}$ to $\mathcal{L}$}
\ENDIF
\ENDIF
\ENDWHILE
\RETURN $\mathcal{G}$
\end{algorithmic}
\end{algorithm}

With the notation above, the creative telescoping algorithm 
from~\cite{Bostan_2018} applies verbatim. It is given in \Cref
{AlgoCT}. Its principle is to
iterate on every monomial of the form
$\partial_1^{\alpha_1}\dots \partial_m^{\alpha_m}$ by increasing order
for some monomial order, e.g., the grevlex order, and to compute the
reduced rational functions $R_{\boldsymbol{\alpha}}$ such that
\begin{equation}
\label{DecMonome2}
	\boldsymbol{\partial}^{\boldsymbol{\alpha}}= R_{
	\boldsymbol{\alpha}}\gamma + \Delta_n G_{
	\boldsymbol{\alpha}} \bmod\mathcal I.
\end{equation}
The rational function $R_{\boldsymbol{\alpha}}$ is obtained by
\begin{equation}\label{red-R}
\begin{cases}
R_{0,\dots,0}=[A^*_1(1)],\\
R_{\boldsymbol{\alpha}}=[\varphi_i(R_{\boldsymbol{\beta}})]\quad\text{
if }\boldsymbol{\partial}^{\boldsymbol{\alpha}}=\partial_i
\boldsymbol{\partial}^{
\boldsymbol{\beta}},
\end{cases}
\end{equation}
where $A_1^*$ is the adjoint of the operator $A_1$ verifying $1=A_1(\gamma)$.
When a monomial~$\boldsymbol\alpha$ is dealt with, two situations are
possible.
The corresponding~$R_{\boldsymbol{\alpha}}$ can be a linear
combination of the previous~$R_{\boldsymbol{\beta}}$.
In that case, that linear combination makes the corresponding
linear combination of~$\boldsymbol{\partial}^{\boldsymbol{\alpha}}$ and the
$\boldsymbol{\partial}^{\boldsymbol{\beta}}$
a newly discovered element of the telescoping ideal $
\mathcal{T}_\mathcal I$ and then it is not necessary to
visit the multiples of this monomial. Otherwise, $
\boldsymbol{\partial}^{\boldsymbol{\alpha}}$
is free from the previous ones and thus a new generator of~$\mathbb B$
has been found. The algorithm terminates when there are no more
monomials to visit.

The only difference with the differential case lies in the definition
of the canonical form $[\cdot]$ associated to the adjoint $L^*$ of the
minimal-order operator~$L\in\mathbb K(n)\langle S_n\rangle$
annihilating~$\gamma$. By \cref{propCT,propCT2} and the
definition of canonical form, \cref{DecMonome2} is satisfied and 
the following equivalence holds:
\begin{equation}
a_{\underline{1}}R_{\underline{1}} + \dots + a_{\underline{s}}R_{
\underline{s}} = 0 \quad \text{ iff } \quad a_{\underline{1}}D^{
\underline{1}} + \dots + a_{\underline{s}}D^{\underline{s}} \in 
\mathcal{T}_{\mathcal I}.
\end{equation}
The following result follows, with the same proof as in \cite{Bostan_2018}.
\begin{theorem}
Given as input the generators of a D-finite ideal $\mathcal I$ and
a cyclic vector~$\gamma$ for~$S_n$, 
Algorithm \ref{AlgoCT} terminates if and only if $\mathcal{T}_
{\mathcal I}$ is D-finite. Then, it outputs a Gröbner basis of $
\mathcal{T}_{\mathcal I}$ for the grevlex order.
\end{theorem}
As in the differential situation \cite{Bostan_2018}, one can modify
the algorithm to compute
all elements of~$\mathcal T_{\mathcal I}$ up to a given bound
on their degree, or to return as soon as one telescoper is found, thus
allowing to recover a generating family of a
sub-ideal of $\mathcal{T}_F$.

\section{Generalized Abramov-Petkov\v sek decomposition}
\label{sec:reduction}
The main contribution of
this article is an algorithm for the computation of canonical
forms as in
\cref{CanForm} for the operator 
\[L^*= \sum_{i=0}^rp_i(n)S_n^{-i}\] 
with polynomial coefficients~$p_i\in\mathbb K[n]$. The
modified Abramov-Petkov\v sek decomposition~\cite{ChenHuangKauersLi2015} is a
special case of this reduction when $L$ has order 1 and once the shell~\cite{ChenHuangKauersLi2015} has been removed~\cite[Sec.~3.5.3]{Bostan_2018}.

The starting point is a decomposition of any rational
function~$R\in\mathbb K(n)$ in the form
\begin{equation}\label{eq:decomp-R}
R(n)= P_\infty(n) + \sum_{i,h} \frac{c_{i,h}(n)}{Q_i(n+h)^{\ell_
{i,h}}},
\end{equation}
with $\ell_{i,h}\in\mathbb N^*$, polynomials $P_\infty,Q_i$ and $c_
{i,h}$ in $\mathbb K[n]$ such that $\deg c_{i,h} < \ell_{i,h}\deg
Q_i$ and
$\gcd(Q_i(n+k),Q_j(n))=1$ for all~$k\in\mathbb Z$ when $i\neq j$. This
is discussed in \cref{sec:shiftless}. 

The vector spaces 
\[\mathbb{K}_{Q_i}(n)\stackrel{\text{def}}
{=}\Vect_\mathbb{K}\left( \frac{n^\ell}{Q_i(n+h)^j} \mid h\in
\mathbb{Z},
j\in\mathbb{N}^*, \ell< j\deg (Q_i) \right)\]
are in direct sum for distinct~$Q_i$ and are left invariant by $L^*$ modulo $\mathbb{K}[n]$.
This allows the reduction
algorithm to operate in each of the $\mathbb K_{Q_i}(n)$
independently. This is described in \cref{subsec:weak-polar,subsec:strong-polar}, before the reduction of
the remaining polynomial part in \cref{subsec:weak-poly,subsec:strong-poly}.

\begin{notation} For two integers $a,b$ with $a\le b$, we write
$\llbracket a;b\rrbracket$ for the set $\{a,a+1,\dots,b\}$.
\end{notation}

\subsection{Decomposition of rational functions}\label{sec:shiftless}
Recall that a polynomial~$Q$ is \emph{square-free} when it does not
have
multiple nontrivial factors. It is \emph{shift-free} when $\pgcd(Q
(n),Q(n+k))=1$ for all $k\in\mathbb Z^*$.

A
\emph{shiftless decomposition}
of a polynomial~$Q$ is a factorization of the form
\[
Q=\prod_{i=1}^v{\prod_{j=1}^{n_i}Q_i(n+h_{i,j})^{e_
{i,j}}},
\]
with $e_{i,j}\in\mathbb{N}^*, h_{i,j}\in\mathbb{Z}$,
and $Q_i\in\mathbb{K}[n]$ are such that each $Q_i$ is square-free and
$\pgcd(Q_i(n+k),Q_j(n))=1$ for all~$i,j$ and all $k\in\mathbb{Z}$
unless $i=j$ and $k=0$.
Such a factorization can be computed using only gcds, resultants
and
integer root finding~\cite{DBLP:conf/issac/GerhardGSZ03}.

Note that shiftless decompositions are not unique in general. One can
be
refined when a $Q_i$ is not irreducible, by splitting this
factor further. In particular, the linear factors of the~$Q_i$ can be
isolated and dealt with more easily.

A polynomial~$Q$ is \emph{refined with respect to a polynomial}~$P$
when it is such that for each
$h\in\mathbb Z$, there exists~$\ell\in\mathbb N$
such that $\gcd(P,Q(n+h)^{\ell+1})=Q(n+h)^\ell$.
A shiftless decomposition is called \emph{refined
with respect to $P$} when each $Q_i$ is.
This refinement can be computed using gcds only and will be used with
$P=p_0$ and $P=p_r$, the extreme coefficients of $L^*$.

From a shiftless decomposition, the partial fraction decomposition of 
\cref{eq:decomp-R} is then obtained by standard algorithms~\cite 
[5.11]{GathenGerhard2013}.

\subsection{Weak reduction of the polar part}\label{subsec:weak-polar}

\begin{algorithm}[t]
\caption{Weak reduction of poles $[\cdot]_{Q}$}
\begin{algorithmic}
\label{WRpolesR}
\REQUIRE $R,Q$
\ENSURE a reduced form of $R$

\WHILE{there exists $j<0$ (minimal) such that $Q(n-j) \mid \denom(R)$}
\STATE $R \gets R - L^*\left(\frac{A(n)}{Q(n-j)^{\ord_j(\denom
(R))+\ord_{j}(p_0)}}\right)$  with $A$ as in \cref{WRvalAbis}
\ENDWHILE

\WHILE{there exists $j\geq r$ (maximal) such that $Q(n-j) \mid \denom(R)$}
\STATE $R \gets R -  L^*\left(\frac{A'(n+r)}{Q(n-j+r)^{\ord_j(\denom(R))+\ord_{j}(p_r)}}\right)$  with $A'$ as in \cref{WRvalAbisp}
\ENDWHILE

\RETURN $R$
\end{algorithmic}
\end{algorithm}

\begin{lemma}
\label{WRlemma1bis}
Let $Q\in\mathbb K[n]$ be square-free, shift-free and refined with
respect to the
coefficients $p_0$ and $p_r$ of~$L^*$. Given
a rational
function $R\in
\mathbb{K}_{Q}
(n)$,
\cref{WRpolesR} computes a rational
function $[R]_{Q}\in\mathbb K_Q(n)$ with all its poles at zeros of $Q
(n-j)$ such
that $j\in\llbracket0;r-1\rrbracket$
and $R - [R]_{Q} = P+L^*
(T)$ for some $P\in\mathbb K[n]$ and $T\in\mathbb{K}_{Q}(n)$. 
The algorithm is
$\mathbb{K}$-linear.
\end{lemma}
\begin{proof} 
Assume that $R$ decomposes as
\begin{equation}
\label{RRdec}
R=\sum_{j\in J} \frac{\lambda_j(n)}{Q(n-j)^{s_j}}\qquad
\text{with}\quad
\deg(\lambda_j(n)) < s_j\deg(Q).
\end{equation}
Let $j_m=\min(J)$ and $\ord_{j}(p_0)$ be the largest integer $\ell$
such that $Q(n-j)^\ell \mid p_0$. Then,
\[
{Q(n-j_m)^{s_{j_m}}}
L^*\left(\frac{1}{Q(n-j_m)^{s_{j_m}+\ord_{j_m}(p_0)}}\right)=
  \tilde{p}_0(n) \bmod Q(n-j_m)^{s_{j_m}},
\]
where $\tilde{p}_0(n)$ is the remainder in the Euclidean division of
$p_0/Q(n-j_m)^{\ord_{j_m}(p_0)}$ by $Q(n-j_m)^{s_{j_m}}$. The poles of
this rational function are at zeros of $Q(n-j)$ with $j\in \hat J:=
J\setminus\{j_m\}\cup(j_m+\llbracket1,r-1\rrbracket)$.

Since $Q$ is reduced with respect to~$p_0$, the polynomial~$\tilde
p_0(n)$ is relatively prime with~$Q(n-j_m)$. Thus, there exist
polynomials $A$ and $B$ such that
\begin{equation}
\label{WRvalAbis}
\lambda_{j_m}(n)=A(n)\tilde p_0(n) + B(n)Q(n-j_m)^{s_{j_m}}.
\end{equation}
Then
\[\frac{A(n)\tilde p_0(n)}{Q(n-j_m)^{s_{j_m}}}=
\frac{\lambda_{j_m}(n)}{Q(n-j_m)^{s_{j_m}}}-B(n),\]
so that
\[
R-  L^*\!\left(\frac{A(n)}{Q(n-j_m)^{s_{j_m}+\ord_{j_m}
(p_0)}}\right)
\]
is equivalent to~$R$ modulo $L^*(\mathbb K_Q(n))$ and with all its
poles at zeros of $Q(n-j)$ with $j\in \hat J$. This operation can be repeated a finite number of times
until all poles are at zeros of $Q(n-j)$ with $j\ge0$.

Similarly, let $j_M=\max(J)$ and $\ord_{j_M+r}(p_r)$ be the largest
integer $\ell$ such that $Q(n-j_M)^\ell$ divides $p_r(n-r)$. Then
\[
{Q(n-j_M)^{s_{j_M}}}
L^*\!\left(\frac{1}{Q(n-j_M+r)^{s_{j_M}+\ord_{j_M+r}(p_r)}}\right)=
  \tilde{p}_r(n) \bmod Q(n-j_M)^{s_{j_M}},
\]
where $\tilde{p}_r(n)$ is the remainder in the Euclidean division of
$p_r/Q(n-j_M)^{\ord_{j_M+r}(p_r)}$ by $Q(n-j_M)^{s_{j_M}}$. The poles
of
this rational function are at zeros of $Q(n-j)$ with $j\in \hat J':=
J\setminus\{j_M\}\cup(j_M-\llbracket1,r-1\rrbracket)$.
Again, since $Q$ is reduced with respect to~$p_r$, the
polynomial~$\tilde p_r$ is relatively prime with $Q(n-j_M)$. Thus
there exist two polynomials $A'$ and $B'$ such that
\begin{equation}
\label{WRvalAbisp}
\lambda_{j_M}(n)=A'(n)\tilde p_r(n) + B'(n)Q(n-j_M)^{s_{j_M}}
\end{equation}
so that 
\[
R -  L^*\left(\frac{A'(n+r)}{Q(n-j_M+r)^{s_{j_M}+\ord_{j_M+r}
(p_r)}}\right)
\]
is equivalent to~$R$ modulo $L^*(\mathbb K_Q(n))$ and with all its
poles at zeros of $Q(n-j)$ with $j\in \hat J'$. This operation can be repeated a finite number of times
until all poles are at zeros of $Q(n-j)$ with
$j\in\llbracket0,r-1\rrbracket$.

Each step being~$\mathbb K$-linear, so is the algorithm.
\end{proof}

\begin{example}
\label{expleWR}
Let
\begin{multline*}
R= 8nx+ \frac{n(8x^2 +x) - 16x - 1}{2(n - 1)^2} -\frac{4(x - 1)x^2}{n} \\
 - \frac{(4x^3 + 8x^2 - 31x - 32)n + 4x^3 - 31x - 32}{2(n+1)^2} +\frac{4x(x^2 - x - 8)}{n + 2} + \frac{2x^3}{n + 3}
\end{multline*}
and
\[
L^*=x^2(n-2)S_n^{-3} -n(4n^2 - x^2 - 4n)S_n^{-2} + n(4n^2 - x^2 - 4n)S_n^{-1} - x^2(n+2).
\]
The poles of $R$ are at $\{1,0,-1,-2,-3\}$. We take $Q=n+1$
and follow the steps of the algorithm.
 
The pole at $-3$ is easy: from
\[
L^*\!\left(\frac{1}{n+3}\right)=-x^2 + 4n + \frac{-x^2 + 8}{n + 1} + 
\frac{2x^2 - 16}{n + 2} + \frac{x^2}{n + 3}
\]
and the coefficient $2x^3$ of $(n+3)^{-1}$ in $R$, the algorithm
performs the subtraction
\[
R \gets R - 2xL^*\!\left(\frac{1}{n+3}\right)
 = \frac{n(8x^2 +x) - 16x - 1}{2(n - 1)^2} + \frac{4x^2}{n}  - \frac{(8x^2+ x - 32)n + x - 32}{2(n + 1)^2} -\frac{4x^2}{n + 2}.
\]
Next, the pole $-2$ is a simple root of the constant coefficient of
$L^*$, leading to the computation of 
\[
L^*\!\left(\frac{1}{(n+2)^2}\right)= \frac{x^2(n - 2)}{(n - 1)^2} + 
\frac{x^2}{n} - \frac{n(x^2 - 4) - 4}{(n + 1)^2} - \frac{x^2}{n + 2}
\]
so that the pole is removed by
\begin{equation}
R \gets R - 4L^*\!\left(\frac{1}{(n+2)^2}\right)
= \frac{x}{2(n - 1)} - \frac{x}{2(n + 1)}.\label{eq:ex-lasteq}
\end{equation}
$R$ now has all its poles in $\{-1,0,1\}$ and the weak reduction
is finished.
\end{example}

\subsection{Strong reduction of the polar part}
\label{subsec:strong-polar}
By \cref{WRlemma2}, 
the weak reduction produces rational functions all whose poles
differ from those of~$Q$ by an integer in $\llbracket
0,r-1\rrbracket$.
The next step of the reduction is to subtract rational functions in
$L^*(\mathbb K_Q(n))$ that have this property.

It turns out to be possible to focus on a finite-dimensional
subspace of 
$L^*(\mathbb{K}_Q(n))$ thanks to the following.
\begin{lemma}
\label{WRlemma2bis} 
If $j< 0$, $s > \ord_{j}(p_0)$ and $\ell<s\deg(Q)$ or
if $j \geq 0$, $s > \ord_{j+r}(p_r)$ and $\ell<s\deg(Q)$ 
then 
\[\left[L^*\!\left(\frac{n^\ell}{Q(n-j)^s}\right)\right]_{Q}=0.\]
\end{lemma}
\begin{proof}
Let $j,s,\ell$ be three integers  that satisfy the first
assumption. Then $L^*({n^\ell}/{Q(n-j)^s})$ has a denominator
that is divisible by $Q(n-j)$ with $j<0$ by assumption. No smaller~$k$
is such that $Q(n-k)$ divides the denominator.
Thus the first pass through the first loop
of the weak reduction subtracts $L^*({n^\ell}/{Q(n-j)^s})$ to itself and reduces it to zero. 
When the second assumption is satisfied, then $L^*({n^\ell}/{Q(n-j)^s})$ has a denominator
that is divisible by $Q(n-(j+r))$ with $j+r\ge r$ by assumption. No
larger~$k$ is such that $Q(n-k)$ divides the denominator. Thus again,
the second loop reduces that fraction to~0.
\end{proof}
\begin{corollary}\label{WRlemma2} 
Let $I_0:=\{j\in\mathbb Z_{<0}\mid \gcd(p_0(n),Q(n-j))\neq1\}$
and $I_r:=\{j\in\mathbb Z_{\ge0}\mid \gcd(p_r(n),Q(n-j))\neq1\}$.
The $\mathbb K$-vector space $[L^*(\mathbb K_Q(n))]_Q$ is
generated by the fractions
\[\left[L^*\!\left(
\frac{n^\ell}{Q(n-j)^{s_j}}\right)\right]_Q,\qquad\text{with}\quad
\begin{cases}
j\in I_0\text{ and
 }1\le s_j\le\ord_j(p_0)\text{ and }0\le\ell<s_j\deg Q\\
 \text{ or }\\
j\in I_r\text{ and
 }1\le s_j\le\ord_{j+r}(p_r)\text{ and }0\le\ell<s_j\deg Q.
\end{cases}
\]
\end{corollary}
\Cref{WRlemma2} gives a generating family of the finite dimensional $\mathbb K$-vector space $[L^*(\mathbb K_Q(n))]_Q$. 
These rational functions can be written in the basis $(n^i/Q
(n-j)^k)_{i,k\in\mathbb N,j\in \mathbb Z}$ and one can then compute an
echelon basis of this finite-dimensional space.
This precomputation step corresponds to the computation of the $B_{Q_i}$'s in \cref{AlgoPrecomputation}.
 The \emph{strong reduction} of a rational
function~$R\in\mathbb K_Q$ then consists in reducing $
[R]_Q$ with this echelon basis. 
By this process, we obtain the
following.
\begin{proposition}
\label{SRlemma}
Strong reduction reduces every rational function $R\in L^*(\mathbb{K}_Q(n))$ to a polynomial in $\mathbb{K}[n]$.
\end{proposition}
\begin{example}\label{expleSR}
With the same notation as in example \ref{expleWR}, \cref{WRlemma2}
shows that $[L^*(\mathbb{K}_{n+1}(n))]_{n+1}$ is generated by 
\[
\left[L^*\left(\frac{1}{n-2}\right)\right]_{n+1}=4n + \frac{x^2}{n +
1}
- \frac{x^2}{n - 1},\qquad
\left[L^*\left(\frac{1}{n-1}\right)\right]_{n+1}=4n + \frac{x^2}{n -
1} - \frac{x^2}{n + 1}.
\]
Thus the strong reduction of the rational function~$R$ from 
\cref{eq:ex-lasteq} is the polynomial
\[
R+ L^*((n-2)^{-1})/(2x)=-2n/x,
\]
concluding the reduction.
\end{example}

\subsection{Weak reduction of polynomials}\label{subsec:weak-poly}
\begin{algorithm}[t] 
\caption{Weak reduction of polynomials $[\cdot]_\infty$}
\begin{algorithmic}
\label{WRpol}
\REQUIRE $P$ and ($\sigma$,$p$) from \cref{eq:indpolinfty}
\ENSURE a reduced form of $P$
\STATE $a \gets 0$
\WHILE{$\deg(P) \geq \sigma$}
\STATE\algorithmicif{ $\deg(P)-\sigma$ is a root of $p$}
\algorithmicthen\ $a \gets a + \operatorname{lt}(P)$; $P\gets P - 
\operatorname{lt}(P)$
\STATE\algorithmicelse\ 
$P \gets P - \frac{\dom(P)}{p(\deg(P)-\sigma)}L^*(n^{\deg(P)-\sigma})$
\ENDWHILE
\RETURN  $a+P$
\end{algorithmic}
\end{algorithm}
The weak reduction of polynomials is a direct adaptation of the differential case 
\cite{Bostan_2018}.
The indicial polynomial of $L^*$ at infinity is the polynomial $p\in
\mathbb{K}[s]$ defined by
\begin{equation}\label{eq:indpolinfty}
L^*(n^s)=n^{s+\sigma}(p(s) + O(1/n)),
\end{equation}
with $\sigma\in\mathbb{N}$.
The ensuing weak reduction is presented in \cref{WRpol}.

\begin{example}
\label{expleSRP}
In \cref{expleWR,expleSR},
the indicial equation at infinity is 
\[
L^*(n^s)=n^{s+2}(8+4s +O(1/n)).
\]
The polynomial $-2n/x$ found in \cref{expleSR} cannot be reduced
further by weak reduction since its degree in~$n$ is smaller than~2.
\end{example}
The following properties are proved exactly as those for weak
reduction at a pole.
\begin{lemma}
\label{WRpollemma1}
Algorithm  $[\cdot]_\infty$ terminates and is $\mathbb{K}$-linear.
For all $P\in\mathbb{K}[n]$, there exists $Q\in\mathbb{K}[n]$ such that $P-[P]=L^{*}(Q)$.
If $s\in\mathbb{N}$ is not a root of  $p$, then $[L^*(n^s)]_\infty=0$.
\end{lemma}

\subsection{Strong reduction of polynomials}\label{subsec:strong-poly}
The final step is to subtract polynomials in $L^*(\mathbb K(n))$. Here
again, a finite number of generators can be obtained thanks to the
following.

\begin{lemma}
\label{SRlemmapol}
Let $Q_1,\dots,Q_v$ be the polynomials that occur in a
shiftless decomposition of $p_0p_r$ and let $P$ be a polynomial in
$L^*(\mathbb{K}(n))$. Then
\[ 
P \in E_{\text{pol}}\stackrel{\text{def}}{=}L^*(\mathbb{K}[n])+\sum_
{i=1}^v [L^*(\mathbb{K}_{Q_i}
(n))]_{Q_i}\cap\mathbb{K}[n] .
\]
\end{lemma} 
\begin{proof}
If $R\in\mathbb K(n)$ is such that $L^*(R)$ is a polynomial, then the
poles of~$R$ must be cancelled by the zeros of $p_0p_r$ or their
shifts.
It follows that $R$ decomposes as
\[
R= R_\infty + \sum_{i=1}^{v} R_i
\]
with $R_i\in\mathbb{K}_{Q_i}(n)$ and $R_\infty\in\mathbb K[n]$. Each
$L^*(R_i)$ has to be a polynomial and thus invariant by $[\cdot]_
{Q_i}$. This concludes the proof.
\end{proof}
By \cref{SRlemmapol}, the vector space $[L^*(\mathbb K[n])]_\infty$ is
generated by $\{[L^*(n^s)]_{\infty}\mid s\in\mathbb N\text{ and }p(s)=0\}$ with
$p$ the indicial polynomial of $L^*$ at infinity.
Generators of each $[L^*(\mathbb K_{Q_i})]_{Q_i}\cap\mathbb K
[n]$ are obtained from the echelon basis used in the strong reduction
with respect to~$Q_i$. This gives a finite set of generators for $[E_
{\text{pol}}]_\infty$, which is easily transformed into a basis by a row
echelon computation. Strong reduction consists in reducing modulo this
basis. The following consequence is as in the polar case.
\begin{lemma}
\label{SRPlemma}
The strong reduction of polynomials reduces every polynomial $P\in L^*(\mathbb{K}(n))$ to zero.
\end{lemma}
\begin{example}
Continuing \cref{expleWR,expleSR,expleSRP},
the polynomial $p(s)=8+4s$ has no positive integer root therefore  $
[L^*(\mathbb{K}[n])]_\infty=\{0\}$. A basis of $[L^*(\mathbb{K}_
{n+1})]_{n+1} \cap \mathbb{K}[n]$ is $\{n\}$ according to example 
\ref{expleSR}.
Therefore $2n/x$ reduces to 0.
\end{example}

\subsection{Canonical Form}

\cref{AlgoPrecomputation} and \Cref{AlgoRedbis}  combine the previous algorithms 
to produce a canonical form.

\begin{algorithm}[t]
\caption{PrecomputeBases}
\begin{algorithmic}
\label{AlgoPrecomputation}
\REQUIRE $Q_{1},\dots,Q_{v}$ polynomials that occur in the shiftless decomposition of $p_0p_r$
\ENSURE The echelon bases $B_{Q_1},\dots,B_{Q_v},B_{pol}$
\STATE $B_{pol} \gets \{\}$
\FOR {$i=1$ \TO $v$ }
\STATE $B_{Q_i} \gets$ Echelon$\left(\left\{\left[L^*\left(\frac{n^l}{Q(n-j)^s}\right)\right]_{Q_i}  \mid l,j,s_j \text{ as in Cor.1 }\right\}\right)$
\STATE $B_{pol} \gets B_{pol}\cup(B_{Q_i}\cap\mathbb K[n])$
\ENDFOR
\STATE $B_{pol} \gets$ Echelon($B_{pol}\cup \{[L^*(n^s)]_\infty \mid s$ integer root of $p \})$

\RETURN  $B_{Q_1},\dots,B_{Q_v},B_{pol}$
\end{algorithmic}
\end{algorithm}

\begin{algorithm}[t]
\caption{Reduction of rational functions $[\cdot]$}
\begin{algorithmic}
\label{AlgoRedbis}
\REQUIRE $R$ and $B_{Q_1},\dots,B_{Q_v},B_{pol}$ computed by Algorithm
\ref{AlgoPrecomputation}
\ENSURE a reduced form of $R$ by $L^*(\mathbb{K}(n))$.
\STATE Decompose $R$ as $P_\infty + \sum_{i=1}^v R_i$ with
$R_i\in\mathbb K_{Q_i}$
\STATE\algorithmicfor\ {$i=1$ \TO $v$ }\algorithmicdo\ 
 $R_{Q_i} \gets \operatorname{StrongReduce}([R_i]_{Q_i},B_{Q_i})$
\STATE $R \gets P_\infty + \sum_{i=1}^s R_{Q_i}$
\STATE Write $R= P + \tilde{R}$ with $P$ a polynomial and $\deg(\tilde{R})<0$
\STATE $P \gets \operatorname{StrongReduce}([P]_\infty,B_{pol})$
\RETURN  $P + \tilde{R}$
\end{algorithmic}
\end{algorithm}

\begin{theorem}
\Cref{AlgoRedbis} computes a canonical form.
\end{theorem}

\begin{proof}
\Cref{AlgoRedbis} is linear as every step is linear. By \cref{SRlemma}
and \cref{SRPlemma}, $[L^*(\mathbb{K}(n)]$ reduces to 0 and $[R]-R \in
L^*(\mathbb{K}(n))$ as only functions in this image were subtracted to
$R$.
\end{proof}

\section{Certificates}\label{sec:certificates}
Reduction-based creative telescoping
algorithms allow to find a telescoper without having to compute
an associated certificate. This has led to faster algorithms as
certificates are known to be larger than telescopers~\cite{bostan:hal-00777675}.
This approach makes sense in the differential
case when it is known in advance that the integral of a certificate
over a cycle that avoids singularities is equal to zero. The framework
is not as favorable for sums. Indeed, it is necessary to detect
whether
the certificate has poles in the range of summation and
it is often unclear whether the certificate becomes~0 at the boundaries of the summation interval.
 
It is however possible to compute the certificates \emph{in a compact
way}
during the execution of our algorithm, with almost no impact on the
execution time. 
The idea
is to make the computation and storage of certificates efficient by
storing them as directed acyclic graphs (dags) rather than
operators with normalized rational function coefficients. These dags
have a number of internal nodes of the same order as the number of
operations performed when computing the telescoper, so that their
computation does not burden the complexity. They can then be evaluated
at the endpoints of the range of summation, or expanded in Laurent
series there. 

\subsection{Computation and structure of certificates}
\Cref{DecMonome2,red-R} show how one
can compute telescopers without computing certificates. We now show
how to compute the certificates simultaneously. These are obtained as
sums of monomials in the~$\partial_i$ multiplied by rational
functions. These rational functions have denominators of the
form~$Q
(n-j)^s$, with integers~$j,s$ and polynomials $Q$ that 
occur either in the shiftless
decomposition of~$p_0p_r$ or in the denominator of the
operator~$A_1\in\hat{\mathbb K}\langle S_n\rangle$ such
that~$1=A_1\gamma$, or in the denominator of one of the
operators~$B_i$ from \cref{propCT}. Those sums are not
reduced to a common denominator. They share many common coefficients
and denominators that are efficiently compacted into dags by sharing
common subexpressions (this is how Maple stores them by default).

The starting point is the cyclic vector $\gamma\in\mathbb A/\mathcal
I$ and the operator $A_1\in\hat{\mathbb K}\langle S_n\rangle$ such
that $1=A_1\gamma$. By Euclidean division by $\Delta_n$ on the left,
\[A_1=R_0+\Delta_n g_0.\]
In general,
the certificate $G_{\boldsymbol\alpha}$ in \cref{DecMonome2} is stored
as an
(unreduced) element $g_{\boldsymbol\alpha}$
of $\mathbb A$ such that $G_{\boldsymbol\alpha}=g_
{\boldsymbol\alpha}\gamma\bmod\mathcal I$. The computation is incremental. It
is initialized with $R_0$ and $g_0$ as above, corresponding to
$\boldsymbol\alpha=\boldsymbol0$. In many cases, the vector~1 is
cyclic, so that one can take~$\gamma=1=A_1=R_0$ and
$g_0=0$. Otherwise, the denominators of~$R_0$ and~$g_0$ are shifts of
the denominator of~$A_1$.

Next, \cref{red-R} leads to the computation 
of the canonical form of the rational function $\phi_i(R_
{\boldsymbol\beta})$. This
computation is performed via a sequence of reductions which consist of
subtractions of elements in $L^*(\mathbb K(n))$. Keeping track of
these rational functions (without normalizing them) gives the
canonical form $R_{\boldsymbol\alpha}$ as
\[
\phi_i(R_{\boldsymbol\beta})=R_{\boldsymbol\alpha}-L^*(c_
{\boldsymbol\alpha})
\]
for some rational function $c_{\boldsymbol\alpha}\in\hat{\mathbb K}$.
Both~$R_{\boldsymbol\alpha}$ and~$c_{\boldsymbol\alpha}$ have
denominators that are shifts of those of~$R_{\boldsymbol\beta}$ or of
factors of~$p_0p_r$.
In view of \cref{eq:def-varphii}, it follows that
\[
\boldsymbol\partial^
{\boldsymbol\alpha}=\partial_i\boldsymbol\partial^
{\boldsymbol\beta}=R_{\boldsymbol\alpha}\gamma+\Delta_n\left(
\partial_iG_{\boldsymbol\beta}+P_{B_i}(\sigma_i(R_
{\boldsymbol\beta}))\gamma+P_L(c_{\boldsymbol\alpha})\gamma\right)\bmod \mathcal I.
\]
Thus, the certificate $G_{\boldsymbol\alpha}$ is obtained as
$g_{\boldsymbol\alpha}\gamma$
with
\begin{equation}\label{eq:certificate}
g_{\boldsymbol\alpha}=\partial_ig_{\boldsymbol\beta}+P_{B_i}(\sigma_i
(R_{\boldsymbol\beta}))+P_L(c_{\boldsymbol\alpha}).
\end{equation}
This proves the claim concerning the structure of the certificates and
the factors of their denominators.  In our implementation, $\partial_i$
is commuted with the coefficients of the certificate~$g_
{\boldsymbol\beta}$ only. If desired, one can further 
use $\partial_i=B_i\bmod \mathcal I$
so as to write $g_{\boldsymbol\alpha}$ as an operator in
$S_n$ only.

\subsection{Evaluation of the certificates}\label{sec:eval-cert}
The output of \cref{AlgoCT} is a set of elements~$T$ of the
telescoping ideal, which means that 
\[T(F)=g(F)(n+1,x_1,\dots,x_m)-g(F)(n,x_1,\dots,x_m),\]
with $g$ a certificate as described above. Summing over~$n$, the
right-hand side telescopes and only the values of the certificate~$g
(F)$ at the endpoints are needed. 

It is possible to prove that these evaluations are
zero without any evaluation in two important cases. First, if the
summand $F$
has finite support (e.g., binomial sums), then the sum of any certificate over $\mathbb{Z}$ will be zero provided it has no pole in the summation range. 
The second case is when one can prove that $R(n)\gamma
(F),\dots,R(n)S_n^{r-1}\gamma(F)$ tend to zero as $n$ tends to  $\pm
\infty$ for any rational function $R\in\mathbb{K}[n]$ (as in the introductory example). Then again the sum of any certificate over $\mathbb{Z}$ will be zero provided it has no pole in the summation range.

\subsection{Integer pole detection}
By its very nature, the method of creative telescoping requires the
certificate not to have poles in the range of summation, so that
telescoping can occur. 
The structure of the certificates described above does not allow the
efficient computation of its denominator exactly. However it is
possible to
compute a multiple of it by taking the least common
multiple of the denominators of every rational function in the
representation. This can be done efficiently by performing the
computation on the dag representation of the certificates.

From this multiple of
the denominator of the certificate, one can compute the set of roots that lie in the summation range;
this amounts to computing the roots that differ from the
endpoints of the summation range by an integer.
If that set is not empty, then one can
compute a Laurent series expansion of the certificate at any
point to check whether it is a pole or not, again by exploiting the
dag
representation of the rational function coefficients.

\subsection{Examples}

\subsubsection{Neumann's Addition Theorem for Bessel functions}
On input $S(x)=\sum_{n=1}^\infty J_n(x)^2$ where $J_n(x)$ is the
Bessel
function of the first kind, \cref{alg:myalg} outputs the
telescoper~$\partial_x$ and a
certificate $G$ in dag form that
exhibits poles at $n\in\{-1,0,1,2\}$. Our implementation produces
the
polynomial $n(n+1)(n-2)(n-1)$ containing this information.
In this example, the certificate is small enough that it can easily be
normalized and one gets its value as
\[-\frac{x }{4 \left(n +1\right)}S_{n}^{2}
+\frac{n
+1}{x}S_{n}-\frac{8 n^{2}-x^{2}+8 n}{4 x \left(n +1\right)},\]
showing that the poles at~1 and~2 vanish in the normalization. Without
normalizing the certificate, one can still evaluate the
series expansions of the certificate at those points to
establish that it
has a finite limit there, making the summation legitimate.
At $n=1$, the evaluation 
is found to be
\[-J_0(x)J_1(x)=\frac12(J_0(x)^2)'.\]
Thus we have proved 
\[\frac{d}{dx}(\frac12J_0^2(x)+J_1(x)^2+J_2(x)^2+\dotsb)=0.\]
This shows that the sum is constant and the value is revealed by its
value at~0, which follows from $J_k(0)=0$ for $k>0$ and $J_0(0)=1$, so
that in the end, we recover the classical identity~\cite[10.23.3]{OlverLozierBoisvertClark2010}
\[1=J_0(x)^2+2\sum_{k\ge1}{J_k(x)^2}.\]

\subsubsection{Apéry's Sequence}
The classical sum
\begin{equation}\label{eq:Apery}
A_n=\sum_{k=0}^n\binom{n}{k}^2\binom{n+k}{k}^2,
\end{equation}
used by Apéry in his proof of the irrationality of~$\zeta(3)$,
has telescoper
\begin{equation}\label{eq:telesc-Apery}
\mathcal T:=\left(n +2\right)^{3} S_{n}^{2}-\left(2 n +3\right) \left
(17
n^{2}+51 n +39\right) S_{n}+\left(n +1\right)^{3}.
\end{equation}
The singularities of the certificate obtained by our implementation
are at~$k\in\{n+1,n+2\}$. Indeed, once normalized, the certificate 
is found to be
\[\mathcal C:=\frac{4 \left(3 k-4 n -8 \right) k^{4}}{\left(k -n
-2\right)^{2}} S_{n}
-\frac{4 k^{4} \left(3 k +4 n +4\right)}{\left(k-n -1 \right)^{2}}.\]
Let $U_{n,k}$ denote the product of binomials in the sum. 
Summing $\mathcal TU_{n,k}$ from~$k=0$ to~$k=n+2$ gives~$\mathcal T
(A_n)$. If telescoping is legitimate, then the values of the endpoints
are the values of $\mathcal CU_{n,k}$ at $k=0$ and $k=n+3$, that are
both easily checked to be~0. For this to allow to conclude that
$\mathcal T(A_n)=0$, it is then sufficient to check that $\mathcal CU_
{n,k}$ is not singular at~$k=n+1$ and~$k=n+2$, even though~$\mathcal
C$ is. Indeed, a series expansion of the
evaluation
of~$\mathcal CU_{n,k}$ at~$k=n+1$ and~$k=n+2$ is
possible for our
implementation, and finds that the sequence has a finite limit
there, which concludes the proof that the telescoper \cref{eq:telesc-Apery}
cancels the sum \cref{eq:Apery}; see also~%
\cite{ChyzakMahboubiSibut-PinoteTassi2014} for more on these issues.

\subsection{A larger example}
The computation of telescoper and certificate for the sum
\begin{equation}\label{eq:larger-example}
\sum_{n=0}^\infty \frac{(4n+1)(2n)!}{n!^2 2^{2n}\sqrt{x}}J_{2n+1/2}(x)P_{2n}(u)
\end{equation}
takes less than~15~sec. with our current implementation (see 
\cref{table:timings}). The telescopers are quite small:
\[(1-u^2)\partial_u + xu\partial_x,\quad
 (u^3 - u)\partial_u^2 + 
 (1+u^2)\partial_u - u^3x^2.\]
In this example, not normalizing the certificates during their
computation has a cost. The actual certificates, once reduced
by the Gröbner basis
of the annihilating ideal of the summand, are not very large. They are
easily computed by Koutschan's program. Still, the corresponding dags are large.
Nonetheless, it takes less than 1~sec. to compute a multiple of
the denominators of the certificates and detect that they do not have
integer roots. Evaluating the certificates at~$n=0$ using their dag
representation takes less than 2~min. and proves that the telescopers
cancel the sum in \cref{eq:larger-example}.

\section{Implementation}
This algorithm is implemented in Maple\footnote{The implementation is
available at
\url{https://github.com/HBrochet/CreativeTelescoping.git}, together
with sessions of examples.}.
\Cref{table:timings} gives a comparison of our code with Koutschan's
heuristic (HF-FCT) and
Chyzak's
algorithm (HF-CT)\footnote{The code was run
on a Intel Core i7-1265U with 32 GB of RAM.}. They
are both implemented in Koutschan's \textsf{HolonomicFunctions}
package in
Mathematica \cite{Koutschan2009}. The column `redctsum' corresponds to
our
algorithm.

These programs have been executed on a list of 21 easy
examples that were
compiled by Koutschan, as well as more difficult ones given in 
\cref{impl:ligne0,impl:ligne1,impl:ligne2,impl:ligne3,impl:ligne4,impl:ligne5,impl:ligne6,impl:ligne6.5,impl:ligne7,impl:ligne8}
below.
\cref{impl:ligne0} comes from recent identities
involving determinants~\cite{AmdeberhanKoutschanZeilberger23},
\cref{impl:ligne1,impl:ligne2,impl:ligne6.5,impl:ligne7} have been chosen because
they looked natural to experiment with,
\cref{impl:ligne4} is a harder example found in Koutschan's list, 
\cref{impl:ligne3} as well as \cref{impl:ligne6} and it special case 
\cref{impl:ligne5}
come from the
classical book of integral and
series by Prudnikov \emph{et al.}~\cite{PrudnikovBrychkovMarichev1986a}, 
and finally \cref{impl:ligne8}
is an example where Koutschan's heuristic does not stop as it does not
guess correctly the form of the ansatz to use~\cite{ChenHouHuangLabahnWang2022}.

\begin{table}
\begin{center}
\begin{tabular}{l | l | c || r}
   &  HF-CT & HF-FCT & redctsum  \\  \hline 
   easy examples  & 6.7s & 7s & 0.9s  \\
   \cref{impl:ligne0}&   101s  &  49s  & 0.8s  \\
   \cref{impl:ligne1} & 52s & 4s & 1.4s \\
   \cref{impl:ligne2} &   62s  & 1.7s & 5.7s \\
   \cref{impl:ligne3} &  4.9s & 1.4s & 10.3s \\
   \cref{impl:ligne4} & 4.9s & 1.4s & 13.5s       \\
   \cref{impl:ligne5} &  1200s  & 13s & 205s\\
   \cref{impl:ligne6} & $>6$h & 108s & 3338s \\
   \cref{impl:ligne6.5} & 1703s & 4.7s & 580s \\
   \cref{impl:ligne7} & $>1$h & 3.2s$(^*)$ & $>1$h \\
   \cref{impl:ligne8} & $>1$h & $>1$h & 0.4s
 \end{tabular}
\end{center}
\caption{Timings. The notation $(^*)$ means that  we could not check
whether the telescopers
returned by HF-FCT were minimal.\label{table:timings}}
\end{table}

\begin{align}
&\label{impl:ligne0} \sum_{j=1}^n \binom{m+x}{m-i+j}c_{n,j} \quad \text{ where $c_{n,j}$ satisfies recurrences of order 2~\cite[p.~6]{AmdeberhanKoutschanZeilberger23}} \\ 
&\label{impl:ligne1} \sum_{n=0}^\infty C_{n}^{(k)}(x)C_{n}^{(k)}(y)
\frac{u^n}{n!}  \\
&\label{impl:ligne2} \sum_{n=0}^\infty J_n(x)C_n^{(k)}(y)\frac{u^n}{n!} \\
&\label{impl:ligne3} \sum_{k=0}^\infty(-1)^k(4k+1)J_{2k+1/2}(w)P_
{2k}(z) \\
&\label{impl:ligne4} \sum_{n=0}^\infty \frac{(4n+1)(2n)!}{n!^2 2^{2n}\sqrt{x}}J_{2n+1/2}(x)P_{2n}(u) \\
&\label{impl:ligne5} \sum_{k=0}^\infty \frac{(b+3/2)_k}{(3/2)_k(b+1)_k}J_k^{(1/2,b)}(x)J_k^{(1/2,b)}(y) \\
&\label{impl:ligne6} \sum_k \frac{(a+b+1)_k}{(a+1)_k(b+1)_k}J_k^{(a,b)}(x)J_k^{(a,b)}(y) \\
&\label{impl:ligne6.5} \sum_{n=0}^\infty P_n(x)P_n(y)P_n(1/2) \\
&\label{impl:ligne7} \sum_{n=0}^\infty P_n(x)P_n(y)P_n(z) \\
&\label{impl:ligne8} \sum_y \frac{4x + 2}{(45x + 5y + 10z + 47)(45x + 5y + 10z + 2)(63x - 5y + 2z + 58)(63x - 5y + 2z - 5)} 
\end{align}
The family $(S_r)$ is defined by~\cite{GillisReznickZeilberger}
\begin{equation}\label{eq:Sr}
S_r = \sum_{k=0}^n \frac{(-1)^k(rn - (r-1)k)!(r!)^k}{(n - k)!^r k!}.
\end{equation}
For any $r$, our algorithm produces a minimal telescoper of
order $r$ and degree $r(r-1)/2$. The timings are reported in \cref{table:timings-Sr}.
It is unclear why the heuristic HF-FCT does not perform well on this
family.

On most of these examples, the main part of the time of the
computation is spent in the
reductions in the call to CanonicalForm in 
\cref{AlgoCT}. For the two similar sums of 
\cref{impl:ligne3,impl:ligne4}, almost half of the time is spent in
\cref{AlgoPrecomputation} performing the reductions needed to compute
the bases for the strong reduction. This step is crucial to ensure
that the minimal order elements in the telescoping ideal are found.

There are cases, like \cref{impl:ligne6} and the family $S_r$,
where the
intermediate rational functions $R_{\mathbf{\alpha}}$ in 
\cref{DecMonome2} become much larger than the telescopers found
after linear algebra on them. In such situations, the direct,
non-incremental approach taken by HF-CT and HF-FCT can be more
efficient, by avoiding an unnecessarily large basis of rational
functions.

\begin{table}
\begin{center}
\begin{tabular}{l | l | c || r}
   &  HF-CT & HF-FCT & redctsum  \\  \hline 
$S_{6}$ & 11s & 64s & 0.4s \\
$S_{7}$ & 32s & 331s & 0.6s \\
$S_{8}$ & 106s & 1044s & 1.0s \\
$S_{9}$ & 325s & 3341s & 2.5s \\
$S_{10}$ & 1035s & $>$1h & 5.7s
\end{tabular}\qquad\qquad
\end{center}
\caption{Timings on the family $S_r$ from \cref{eq:Sr}. 
\label{table:timings-Sr}}
\end{table}

\bibliographystyle{plain}
\bibliography{ReductionBasedCreativeTelescopingForSummation}
\end{document}